\documentclass{bjmc}   

\usepackage[english]{babel}
\usepackage{amsfonts}      
\usepackage{pdfpages}

\usepackage{amsthm}
\usepackage{amssymb}
\usepackage{amsmath}
\usepackage{subcaption}
\usepackage{graphicx}

\doi{https://doi.org/10.22364/bjmc.2024.12.2.04}
\proceedings{12}{2}{176}

\begin{document}

%

\title{Optimizing Administrative Divisions: A Vertex $k$-Center Approach for Edge-Weighted Road Graphs}
\titlerunning{Optimizing Administrative Divisions: A Vertex $k$-Center Approach}  
%
%
\author{Peteris DAUGULIS
}
\authorrunning{Daugulis}   
%
%
\institute{Daugavpils university, Daugavpils LV-5401, Latvia
} 
%
\dedication{ORCID 0000-0003-3866-514X}
\emails{peteris.daugulis@gmail.com}
%
%

\maketitle      

\begin{abstract}
Efficient and equitable access to municipal services hinges on well-designed administrative divisions. It requires ongoing adaptation to changing demographics, infrastructure, and economic factors. This article proposes a novel transparent data-driven method for territorial division based on the Voronoi partition of edge-weighted road graphs and the vertex $k$-center problem as a special case of the minimax facility location problem. By considering road network structure and strategic placement of administrative centers, this method seeks to minimize travel time disparities and ensure a more balanced distribution of administrative time burden for the population. We show implementations of this approach in the context of Latvia, a country with complex geographical features and diverse population distribution.

\end{abstract}


\section{Introduction}

Efficient governance relies on continuously improving the organization of a country. Modernization and optimization efforts enhance service delivery, resource allocation, and citizen satisfaction. Unbiased transparent assessments, data-driven insights, and rigorous analysis should underpin modernization reform proposals in democratic societies. Mathematics and computer science provide valuable tools for modeling administrative systems and processes.
The results of data analysis, modeling, and computation should form the bedrock of political and professional discussions.

Administrative division is a crucial element for organizing national and municipal activities and services. Its effectiveness should be enhanced through mathematical modeling. In this article, we offer a novel approach for solving the first-level territorial division problem as a discrete optimization problem and illustrate possible solutions in the case of Latvia. 

 We assume that a country must be partitioned into a set of \sl territorial units (TU)\rm.\ For an administrative division to be well defined it must be based on a small set of quantitative parameters. In this article, we have chosen time as the single unifying human-centered metric. We consider the time spent traveling to and from the TU center by motor vehicle to be the only quantitative parameter used for formulating the administrative division principles. We view the time citizens spend to comply with their obligations or to receive necessary state or municipal services as a form of taxation. Consequently, we seek to distribute this time burden as uniformly as possible among the population. While complete uniformity in travel times for residents across different locations is unattainable within a minute, it is feasible to minimize the difference between the maximum and minimum travel times, as well as reduce the maximum travel time. This may also minimize differences in workloads involving travels from TU centers for staff of different TUs. The proposed approach aligns with the fundamental principles of administrative-territorial division, namely, homogeneity, effectiveness, efficiency, and fairness (Bradshow and Keating, \cite{bradshow}).  This approach departs intentionally from conventional methods encoded in the current legislature that typically prioritize certain regions, local interests, historical boundaries, or population size. Removing population numbers as a factor determining administrative structure is a necessary step to reverse the detrimental human flow to the capital and traditional centers. An objective is to acknowledge the inherent limitations and biases associated with these criteria and cultivate a more equitable and inclusive approach.

 We consider a country as a network of interconnected urban and rural settlements, linked by roads. This network can be modeled as an undirected edge-weighted \sl road graph,\rm\ where vertices (nodes) represent settlements and edges represent roads connecting them. Edge weights are the minimal travel times determined by speed limits and typical traffic conditions necessary to travel the road between the two endpoint vertices.  This representation captures the connectivity and travel times of the road network, providing a framework for analyzing administrative divisions.
We are interested in a well-defined method for constructing an unbiased administrative division of a country in terms of the road graph - finding centers of TUs and subsets of vertices constituting each TU.

We need basic relevant notions and facts from the graph theory. Let $G=(V,E,w)$ be an undirected graph with a nonnegative edge-weight function $w$. The weight of a path is the sum of the weights of all edges in that path. The distance between two vertices $u\in V$ and $v\in V$, $d(u,v)$, is defined as the weight of a $(u.v)$-path of minimal weight. $d:V\times V\rightarrow \mathbb{R}$ is a metric in $V$. Eccentricity of $v\in V$, $e(v)$, is defined as $\displaystyle\max_{x\in V}d(v,x)$. Radius of $G$ is $r(G):=\displaystyle\min_{v\in V}e(v)$. Diameter of $G$ is $d(G):=\displaystyle\max_{v\in V}e(v)$. Center of $G$ is $Z(G):=G[\mathcal{Z}]$, where $\mathcal{Z}:=\{x\in V|e(x)=r(G)\}$. Periphery of $G$ is $P(G):=G[\mathcal{P}]$, where $\mathcal{P}:=\{x\in V|e(x)=d(G)\}$. Median of $G$ is $M(G):=G[\mathcal{M}]$, where $\mathcal{M}$ is the set of vertices $v$ having minimal total distance $\sum\limits_{x\in V}d(v,x)$.

See (Catini et al.,\cite{catini}) for an example of a computational methodology for analyzing geospatial data networks.

The paper is organized as follows. In Section 2 we describe the data and the computational methods. Results and discussion are given in Section 3. Finally, concluding remarks and future work are drawn in Section 4.

\section{Methods} 

\subsection{Data}

Our data for demonstrating an implementation of the proposed travel time-based administrative division approach and execution of all associated computational tasks is an edge-weighted (partial) road graph of Latvia $\Gamma=(V,E,t)$.  Retrieving information from Google Maps (WEB, \cite{b}),  a globally recognized provider of high-quality geospatial information, we have constructed an undirected edge-weighted graph having 1067 vertices and 1753 edges. Vertices cover the whole territory of Latvia except Riga, the capital of Latvia. Riga is not included in our model since it is a large and densely populated city with a complex road network which is difficult to model as a road graph. A majority of towns or other marked populated areas in Latvia and some major crossroads are chosen as vertices.  Edges correspond to roads connecting the vertices. Road status is not taken into account. We note that not all small settlements and roads are included. The edge-weight function $t:E \longrightarrow \mathbb{R}^{+}$ is the travel time by motor vehicle in minutes between the end vertices as recorded by Google Maps in October-November 2023.

Positions of most settlements are chosen within the dots representing them. For some small settlements, the points are chosen outside their boundaries close to the nearest road intersections. For a few larger towns we have chosen several vertices to represent these towns. We have to assume that a typical error of vertex coordinates is the geographical radius of the dot denoting the town in Google Maps. Thus for a majority of vertices position errors are typically between 10 and 100 meters, depending on the settlement size. The implied travel time errors are less than 1 min - the minimal Google Maps time unit. It means that our coordinate errors can be considered as not affecting the edge weights of the road graph and, consequently, computations involving the road graph.

The geographic coordinates are transformed into the Cartesian coordinates. The set of points on the Earth sphere is orthogonally projected onto the plane approximating this set. This step allows us to simplify computations and visualizations. We note that this may be a novel map projection.

The real non-stop travel time given by Google Maps between two vertices having at least 1 intermediate vertex in their shortest path is less than the sum of travel times for all edges in the shortest path. Since we model travel using edge weights we have to implicitly assume that travel includes stopping in all intermediate vertices. This error increases with the number of intermediate vertices. The fastest routes between points as given by Google Maps often differ from the minimal weight paths in $\Gamma$ because of this reason. These errors affect global computations, such as radius or diameter. Their influence is not significant for local computations, e.g. when the number of TUs is such that the faster routes from each TU center contain at most 5 intermediate vertices. Road conditions and other infrastructure limitations that affect travel choices and time are not incorporated in the model as well. 

\subsection{Computations}

\subsubsection{From centers to  TUs in road graphs - Voronoi cells in edge-weighted graphs}

Suppose we are given a road graph - an undirected edge-weighted graph $G=(V,E,w)$.
We define TU as a pair $(V_{TU},c)$, where $V_{TU}\subseteq V$ and $c\in V_{TU}$ is a distinguished element - the \sl TU center.\rm\  
For any \sl centered partition\rm\ $\textbf{P}=\{(V_1,c_1),...,(V_m,c_m\}$ of $V$, $c_i\in V_i$,  we define radius of the TU $(V_i,c_i)$ as $r(V_i,c_i):=\displaystyle\max_{x\in V_i} d(x,c_i)=e(c_{i})\Big{|}_{G[V_i]}$. We define radius of $\textbf{P}$ $r(\textbf{P}):=\displaystyle\max_{i}r(V_i,c_i)$. 
We think of the $G$-subgraph induced by $V_i$ - $G[V_i]$, as the TU road graph of $V_i$.

Suppose we are given a set of centers $S\subseteq V$. For any $c\in S$ define the vertex subset $V_S(c)$ as the Voronoi cell of $c$ as an element of $S$: 
$$V_S(c)=\{v\in V|d(v,c)\le d(v,c'), 
c'\in S, c'\ne c\}.$$  In other words, $V_S(c)$ contains all vertices for which $c$ is reachable faster than any other center. The set $\{V_S(c)\}_{c\in S}$ is a partition of $V$. We call the centered partition $\textbf{V}(S)=\{(V_S(c),c)\}_{c\in S}$ - \sl the centered Voronoi partition\rm\  for $S$.



Any $S\subseteq V$ defines a Voronoi diagram in $G$ (Erwig, \cite{erwig}; Gawrychowski et al., \cite{gawr}; Hakimi and Labbe, \cite{hakimi}; Melhorn, \cite{melhorn}). Voronoi diagrams have been considered for use in territorial management and planning (WEB, \cite{a}; Ricca et al., \cite{ricca}).

To minimize bias within TUs, we impose an additional constraint that the TU center $c$ lies within the graph center of its Voronoi subgraph in the underlying graph $G$ - $c\in Z(G[V_S(c)])$. This ensures that the TU center is positioned centrally within its assigned region, reducing the potential for bias or unfairness. Consequently, the eccentricity of each TU center $c$ in the TU's Voronoi subgraph is equal to the radius of the TU subgraph.

The crucial step is to set optimization conditions and find a solution set $S_{opt}$ of TU centers. This would trivially imply a partition of $V$ - its Voronoi partition $\textbf{V}(S_{opt})$.


The Voronoi partition approach can also be used by accepting a set of preselected district centers and constructing the TUs as the Voronoi partition without further optimization.

\subsubsection{Optimal centers - a minimax facility location problem}

If the set of TU centers $S$ is chosen then we define TUs uniquely as Voronoi cells of $S$. Each TU graph $G[V_S(c)]$ with the center $c$ can be characterized by its graph radius $r(V_S(c),c)$. Note that for any TU its radius depends on the positions of all other TU centers. 

For a fixed number of TUs ($|S|$) we want to choose the center set $S$ such that $\displaystyle\max_{c\in S}r(V_S(c),c)=r(\textbf{V}(S))$ is minimal. The motivation for this requirement is a drive to minimize the maximal "time burden" of TUs - to find a TU division such that the maximal radius of a TU is as small as possible for a given number of TUs. This would ensure that no TU has a substantially larger or smaller radius than the others. Thus we pose an optimal TU problem as a minimax facility location problem, see (Church and Drezner, \cite{church}) for a recent survey.

Since $d(G)\le 2 r(G)$, the radius minimization implies the diameter minimization which corresponds to the minimization of travel times between any TU vertices.

\subsubsection{A modified vertex $k$-center problem and its application to territorial division} 
We can slightly rephrase the problem. Let $G$ be an undirected graph with a metric $d(\cdot,\cdot)$ defined in $V(G)$. The problem is to find a $k$-set $S\subseteq V(G)$ such that $\displaystyle\max_{v\in V(G)}d(v,S)$, where $d(v,S)=\displaystyle\min_{c\in S}d(v,c)$, is minimal. This problem (\sl the vertex $k$-center problem\rm) first proposed in (Hakimi, \cite{hakimi2}) is known to be NP-hard.  Polynomial time approximation algorithms, such as the 2-approximated algorithms Sh (Shmoys,\cite{shmoys}), Hs (Hochbaum and Shmoys, \cite{hoch}) or Gon (Gonzalez, \cite{gonz}) algorithms, use farthest point clustering as an important step. Other possible approximation steps are local search and random sampling in neighborhoods of vertices.

Our problem of finding a set of graph-centered TU centers minimizing the maximal TU radius amounts to the vertex $k$-center problem in $\Gamma$ with the additional condition requiring TU centers to be centers of the TU subgraphs of the ambient graph,  ensuring that the TU centers are well-placed within their administrative units. Formally, it means to find $S\subseteq V$ for which  $$\min\limits_{\substack{S\subseteq V(G), |S|=k,\\ \forall c\in S:c\in Z(G[V_S(c)]}}\Big(\max\limits_{v\in V(G)}\min\limits_{c\in S}d(v,c)\Big)$$  is achieved.

\subsubsection{Approximation algorithms for the territorial division problem}

We find optimal center sets $S_k\subseteq V(G)$,$|S_k|=k$, in two steps. Below we also mean $S_k$ as a running variable.

\begin{enumerate}
    \item[ \sl Step 1\rm\ ] We find a first approximation of the center set using the greedy farthest clustering (maximin) method as in Sh, Hs, or Gon algorithms, where it is the only step. We start with a random or chosen vertex $\{v_1\}=S_1$. We construct the sequence $S_1,S_2,...,S_k$ by the following loop. At the $i$-th iteration we construct $S_i$ by adding to $S_{i-1}$ a vertex $v$ such that $\displaystyle\min_{c\in S_{i-1}}d(v,c)$ is maximal. After finishing this loop and finding $S_k$ we do an additional substep. We repeatedly update $S_k$ shifting each of its elements to graph centers of their Voronoi cells. This does not increase the radius or the eccentricity of the TU center. Then we recompute the Voronoi cells.  The re-computation does not increase the maximal radius as well, see Proposition \ref{1}. In practice, it always decreases the maximal radius and stabilizes after at most 5 shift-recomputation iterations.
    
    \item[ \sl Step 2\rm\ ] We do cycles of exhaustive local searches in neighborhoods of vertices of $S_k$. In practice, we have used $n$-th neighborhoods with $n\le 30$. At each iteration of the inner loop, we exhaustively search in the $n$-th neighborhood of one element of $S_k$ fixing the other elements of $S_k$. If we find a center which defines a center set $S_{k}'$ with $r(\textbf{V}(S_{k}'))<r(\textbf{V}(S_{k}))$, we update the center set: $S_k\leftarrow S_{k}'$. This step is terminated when there is no improvement after a neighborhood search of all $n$ center vertices. Next, we update $S_k$ again by shifting its elements as described in \sl Step 1.\rm\  
    
    \item[ \sl Step 3 (optional)\rm\ ] After the minimization of the maximal TU radius we can maximize the minimal TU radius while keeping the obtained maximal TU radius constant. This makes the radius values as close as possible without changing the maximal radius value. It focuses on enhancing the balance within the administrative division.  The minimization can be achieved in the same way as in \sl Step 2.\rm\
\end{enumerate}

\begin{proposition}\label{1} Let $G=(V,E,w)$ be an undirected edge-weighted graph with a positive weight function, $S\subseteq V$,  
$\textbf{A}=\{(V_S(c),c)\}_{c\in S}$ - the centered Voronoi partition for $S$.
Then for any centered partition $\textbf{B}=\{(V'_{c},c)\}_{c\in S}$ with $c\in V'_a$ we have $r(\textbf{A})\le r(\textbf{B})$.
    
\end{proposition}

\begin{proof} $r(\textbf{A})=\displaystyle\max_{c\in S}r(V_S(c),c)=\max_{c\in S} \max_{v\in V_S(c)} d(v,c)=\max_{v\in V} d(v,c_{v})$ where $c_v\in S, c_v\in V_i: v\in V_i$. Similarly $r(\textbf{B})=\max\limits_{v\in V} d(v,c'_{v})$ where $c'_v$ where $c'_v\in S, c'_v\in V'_j: v\in V'_j$. For any $v\in V$ we have $d(v,c_{v})\le d(v,c')$ for any $c'\in S$, which implies the statement. \qedsymbol

\end{proof}

We note that moving the TU centers to the subgraph centers and re-computation of the Voronoi partition as well as \sl Step 2\rm\ usually decrease the maximal radius of the division, thus these steps can be considered to be improvements of the classic vertex $k$-center approximation algorithms. The algorithm without center shifts and Voronoi recomputations usually produces divisions with lower radius values which can not be used since they violate our requirement about the TU center. We also note that backtracking and random search in neighborhoods of center vertices require more computation time or give higher radius values. Thus they are not practical choices for the second step.

The shifting of traditional TU centers to the TU graph centers and preserving traditional district borders with or without Voronoi repartitioning without optimization can also be used in territorial reform. 


The necessary software has been developed and computations are performed by MAGMA (Bosma et al, \cite{bosma}). The MetaPost Previewer of Henderson (WEB, \cite{c}) is used for visualization.

\subsubsection{Borders of TUs}

After TUs are defined as sets of road graph vertices, it is necessary to draw borders as polygonal chains separating these sets. Drawing justified borders as polylines requires additional arguments beyond road graphs. We consider this problem to be out of the scope of this article. Nevertheless, we sketch some ideas and use them in visualization.

A preliminary border for each TU can be constructed as the alpha shape of its vertex set (Edelsbrunner et al., \cite{edelsbrunner}). The alpha shape method defines a polygon depending on a real parameter for each TU. These polygons exhibit gaps between them. Additional considerations should be used to split these gaps between the neighboring TUs and delineate the borders.

Alternatively, we can use edges, referred to as \sl cross-edges,\rm\ such that their endpoints belong to different TUs. Geographic midpoints of cross-edges can be considered as the initial data for constructing the TU borders. These borders can be defined as graphs using cross-edge midpoint sets as vertices. For instance, minimal spanning trees or relative neighbor graphs (Toussaint, \cite{touss}) can be used to draw borders. In practice, relative neighbor graphs give visually appealing results except for a small number of isolated edges.

\section{Results and discussion}

\subsection{Properties of the road graph of Latvia}

$\Gamma$ is depicted in Fig.1. Graph-theoretic properties of the existing public road network of Latvia, relevant to the topic of the article, can be computed. The maximal edge weight (the travel time between two vertices) is 37 min, the minimal edge weight is 1 min. The mean edge weight is 10.2 min, and the standard deviation of the edge weight multiset is 5 min. The maximal vertex degree is 9, the degree sequence is $[0,35,284,360,217,112,32,18,4,5]$. The mean vertex degree is 3.3, the standard deviation of the degree multiset is 1.3. The mean shortest path weight is 403 min, the standard deviation of the shortest path weight multiset is 323 min. The modal shortest path weight is 127 min.

The diameter of $\Gamma$ is 554 min, the radius is 282 min. $Z(\Gamma)$ is Ozolaine, Bauska municipality. The periphery $P(\Gamma)$ consists of two vertices - Zalesye, Ludza municipality, and Nida (also the antimedian), Dienvidkurzeme municipality. The median of $\Gamma$ is Kegums, Ogre municipality.

\begin{figure}
    \centering
    \includegraphics[width=1\textwidth]{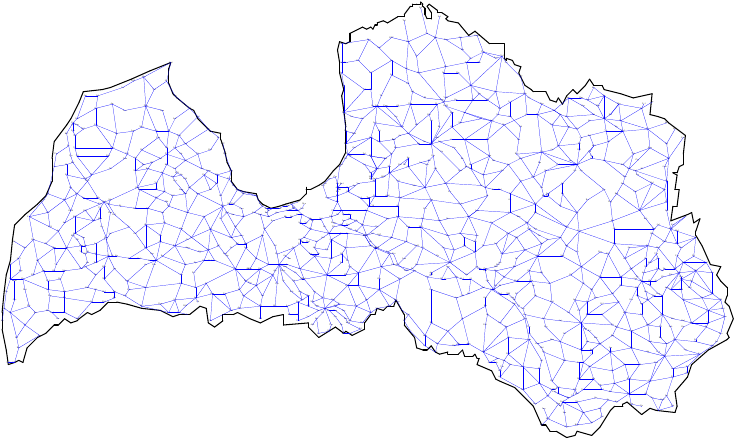}
    \caption{ $\Gamma$ - a partial road graph of Latvia.}
    \label{fig:enter-label}
\end{figure}

\subsection{Examples of territorial division}

\subsubsection{$k=1$}

The case $k=1$ is the problem of finding $Z(\Gamma)$ which can be interpreted as identifying a candidate for the optimal position for the capital of Latvia. $Z(\Gamma)$ is located in Ozolaine, Bauska municipality, as shown in Fig.2.  The eccentricities of vertices near Riga, the capital of Latvia,  are at least 300 min which is bigger than $r(\Gamma)$ =282 min. A transfer of administration functions requiring citizens' travel from Riga to Ozolaine would mean time savings and improved access to services for residents living in the regions of Kurzeme,  Latgale and Zemgale.

\begin{figure}
    \centering
    \includegraphics[width=.7\textwidth]{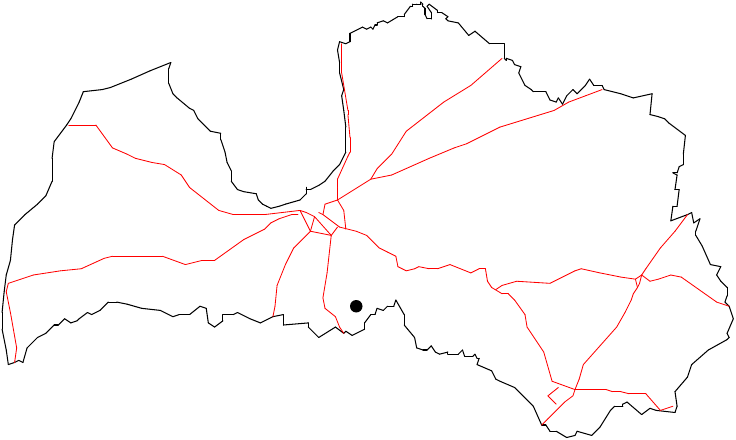}
    \caption{The case $k=1$. The black dot is in the geographic position of $Z(\Gamma)$. Red lines are state main roads.}
    \label{fig:enter-label}
\end{figure}

\subsubsection{$k>1$} We have found instances of TU divisions with low values of the maximal TU radius, denoted $R$, for small values of $k$. For $n=2,3$ we have found the exact minimal values of $R$ by exhaustive search. For some values of $k$, solutions are not unique, but the differences between solutions are local. Table 1 shows the smallest obtained $R$ values for some $k$. The alpha shapes of minimal solutions for $k:\ 2\le k\le 9$ are shown in Fig.3.

\begin{table}
    \centering
    \caption{The smallest obtained $R$ values}
    \medskip
\resizebox{\textwidth}{15pt}
{
    \begin{tabular}{|c|c|c|c|c|c|c|c|c|c|c|c|c|c|c|c|c|c|c|}
    \hline
      k   & 2 &3  &4  &5  &6  &7  &8  &9  &10&11 &12&13&14&15&20&25&30&37\\
      \hline
    R(min)  & 187 & 146 & 132 & 123 & 104 & 97 & 89 & 88  & 81 & 78&75&72&69&66&59&54&51&46\\
    \hline
    \end{tabular}
}
    \label{tab:my_label}
\end{table}

\begin{figure}
\centering
\begin{subfigure}{0.45\textwidth}
    \includegraphics[width=\textwidth]{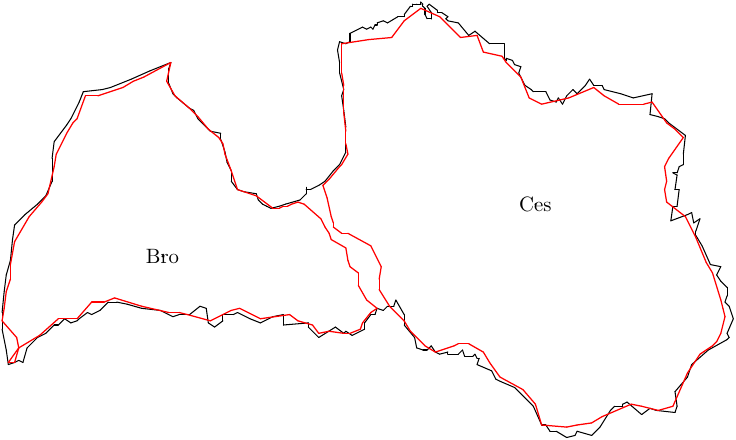}
    \caption{$k=2$.}
    \label{fig:first}
\end{subfigure}
\hfill
\begin{subfigure}{0.45\textwidth}
    \includegraphics[width=\textwidth]{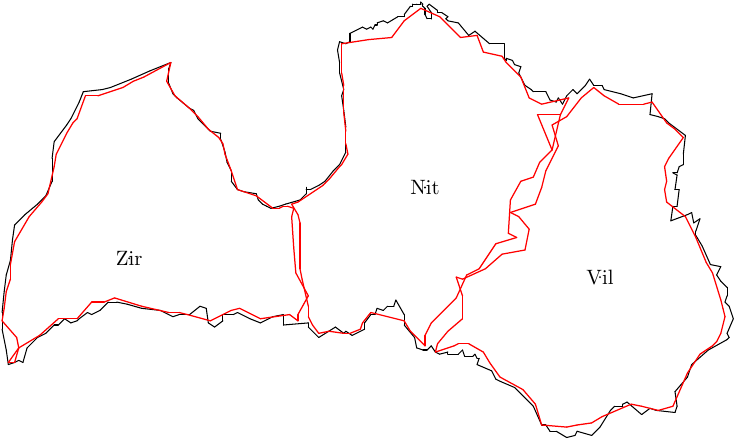}
    \caption{$k=3$}
    \label{fig:second}
\end{subfigure}

\bigskip

\begin{subfigure}{0.45\textwidth}
\centering
    \includegraphics[width=\textwidth]{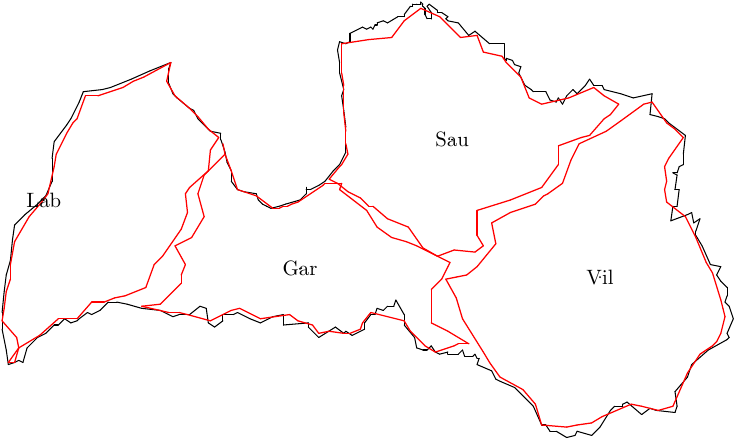}
    \caption{$k=4$.}
    \label{fig:second}
\end{subfigure}
\hfill
\begin{subfigure}{0.45\textwidth}
    \includegraphics[width=\textwidth]{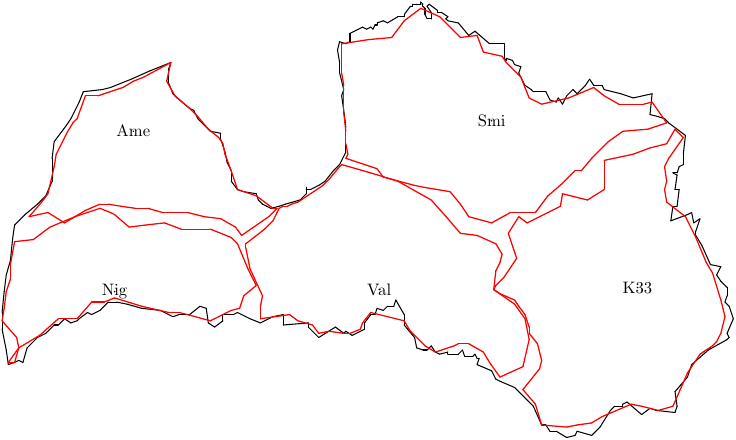}
    \caption{$k=5$}
    \label{fig:second}
\end{subfigure}

\bigskip

\begin{subfigure}{0.45\textwidth}
\centering
    \includegraphics[width=\textwidth]{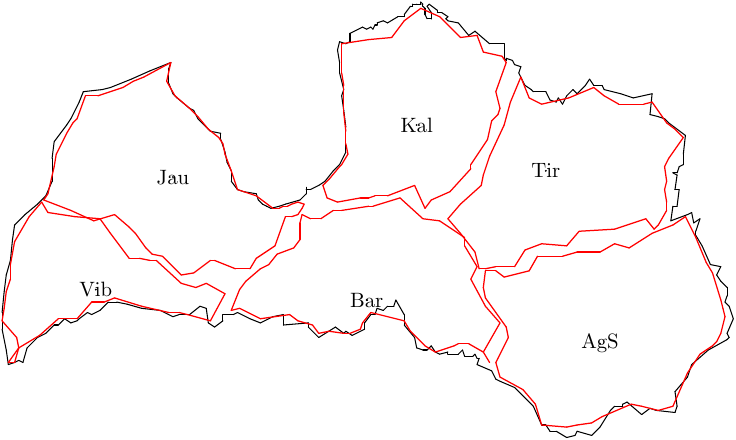}
    \caption{$k=6$.}
    \label{fig:second}
\end{subfigure}
\hfill
\begin{subfigure}{0.45\textwidth}
    \includegraphics[width=\textwidth]{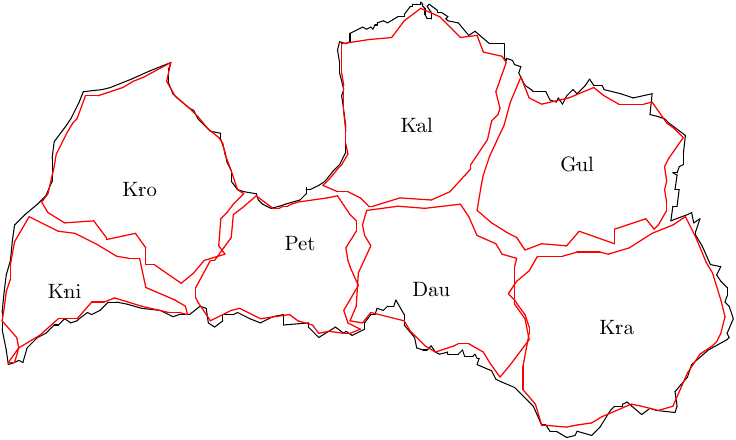}
    \caption{$k=7$}
    \label{fig:second}
\end{subfigure}

\bigskip

\begin{subfigure}{0.45\textwidth}
\centering
    \includegraphics[width=\textwidth]{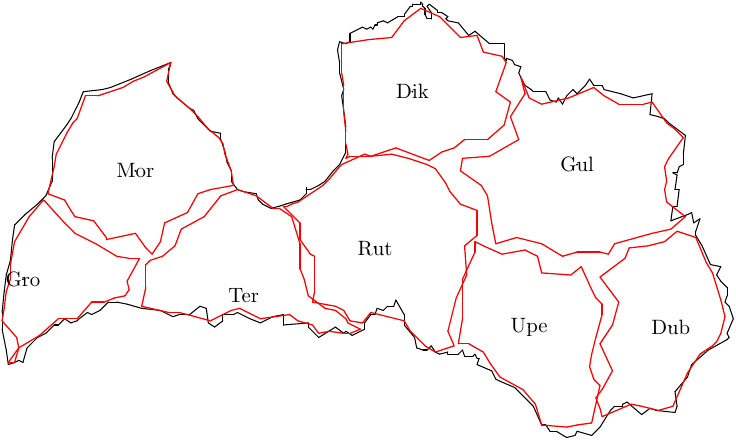}
    \caption{$k=8$.}
    \label{fig:second}
\end{subfigure}
\hfill
\begin{subfigure}{0.45\textwidth}
    \includegraphics[width=\textwidth]{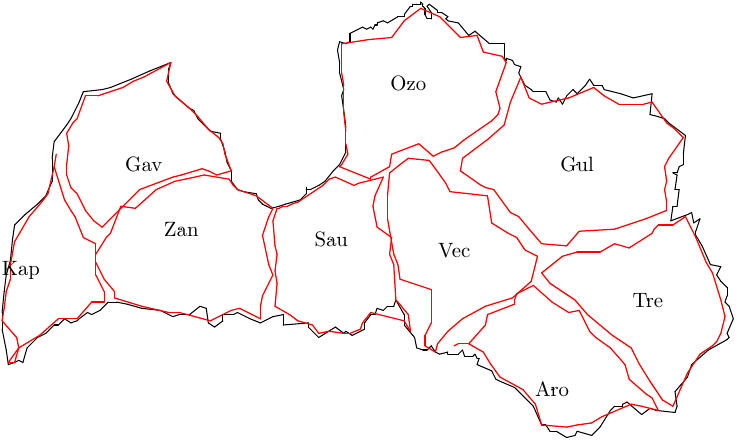}
    \caption{$k=9$}
    \label{fig:second}
\end{subfigure}
        
\caption{Cases $2\le k\le 9$. The alpha shape view of TUs. The letters are at the geographic positions of the centers. Red lines are preliminary borders of TUs.}
\label{fig: 3}
\end{figure}

\bigskip

\noindent
We discuss some special cases of $k$.

\paragraph{$k=5$} Currently, a division of Latvia into 5 historic lands is accepted. We have found an optimal division into 5 TUs with district centers in Amele, Nigrande, A13/A15 roundabout south-west of Rezekne, Smiltene, and Valle, see Fig.4. For this division, the $R$-value is equal to 123 min and the minimal radius is equal to 120 min. Introducing such a territorial division would create TUs having close worst travel time parameters.

\begin{figure}
\centering
\begin{subfigure}{0.45\textwidth}
    \includegraphics[width=\textwidth]{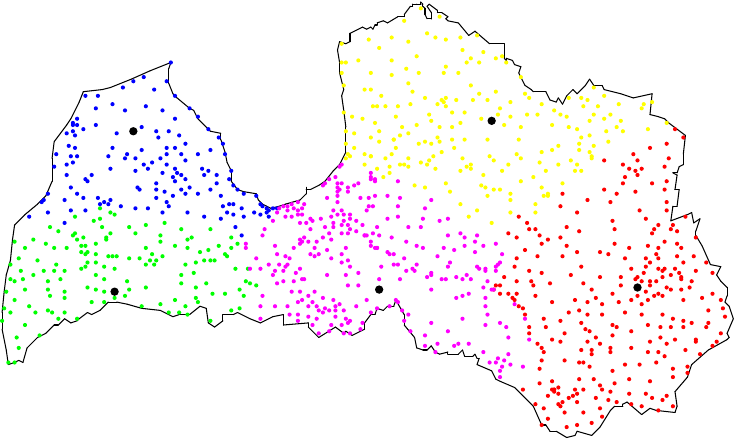}
    \caption{TUs as vertex subsets.}
    \label{fig:first}
\end{subfigure}
\hfill
\begin{subfigure}{0.45\textwidth}
    \includegraphics[width=\textwidth]{F5_1501_123_3.pdf}
    \caption{TU borders as alpha shapes.}
    \label{fig:second}
\end{subfigure}

\bigskip

\begin{subfigure}{0.45\textwidth}
\centering
    \includegraphics[width=\textwidth]{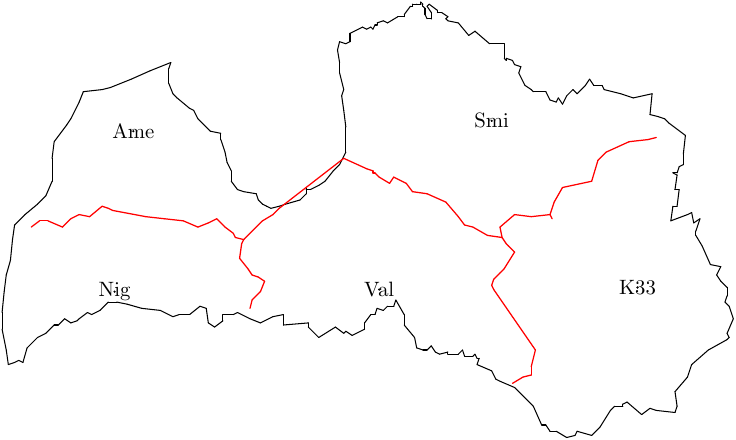}
    \caption{TU borders as relative neighborhood graphs of cross-edge midpoints.}
    \label{fig:second}
\end{subfigure}
        
\caption{A case $k=5$.}
\label{fig: 3}
\end{figure}

\paragraph{$k=15$} For the current first-level administrative division the maximal TU (municipality) radius value is 68 min, see the Dienvidkurzeme municipality. Additionally, the centers of 28 out of 37 TUs are different from their road subgraph centers. In our approach, the minimal number of TUs having the maximal radius less or equal to 68 min is at most 15. A TU division having 15 parts with the maximal radius equal to 67 min is shown in Fig.5. The minimal radius is 60 min, the mean radius is 63.4 min, and the standard deviation is 2.7 min. Since there are 37 territorially separate municipalities, distinct from Riga, in the current territorial division, our method shows a solution that would significantly reduce the maintenance expenses of the municipal system by decreasing the number of TU more than twice.

\begin{figure}
    \centering
    \includegraphics[width=.7\textwidth]{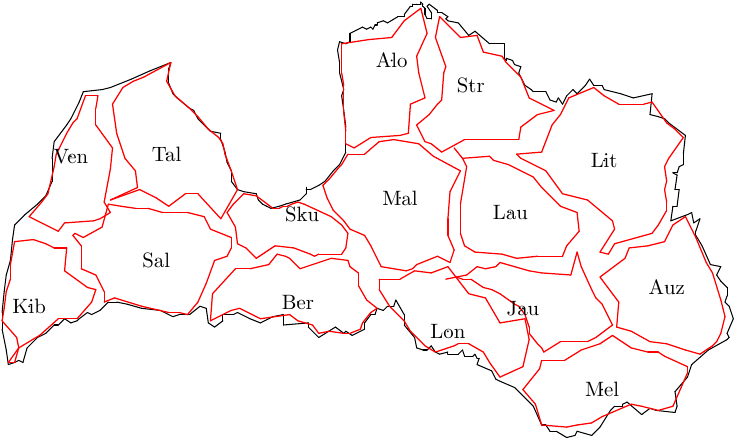}
    \caption{A case $k=15$. The alpha shape view.}
    \label{fig:enter-label}
\end{figure}

\paragraph{$k=37$} For the current administrative first-level division of Latvia the number of TUs can be assumed equal to 37, see \ref{2}.  In our approach, there is a TU division with 37 parts for which the maximal radius is 46 min, the mean radius is 41.6 min and the standard deviation is 3.2 min.  This example also demonstrates the optimization potential of our approach.

\subsubsection{The current first-level administrative division in Latvia and discussion}\label{2}

To provide a comparative context, we examine further the current first-level administrative division in Latvia, as detailed in (WEB, \cite{d}). It comprises 36 municipalities and 7 state cities, including 6 enclaves situated within other municipalities. Consequently, we assume 37 TUs in total. As previously mentioned, Riga, the capital of Latvia, is excluded from our analysis. The maximum TU radius is 68 min, the minimum is 14 min. The mean radius is equal to 42 min, and the standard deviation is 13.5 min. This substantial variation in TU sizes underscores the significant discrepancies in travel times for accessing administrative services among residents. The most frequently observed radius is 47 min, occurring in five TUs. Notably, the administrative center and road subgraph center differ for 28 TUs.

Based on this information and our computations, we conclude that there is a need to reconsider the current first-level administrative division in Latvia. The maximal TU radius should be the initial value for computations. A more efficient system would likely have fewer TUs with boundaries more aligned with the existing road network. Our method would help to reduce travel time inequality and save time for residents to access administrative services.  By consolidating administrative responsibilities into fewer TUs, the government could potentially reduce administrative costs. Enhanced accessibility to administrative services could support economic development in underserved remote rural communities.  Relocating administrative centers can invigorate economic growth and infrastructure development in new areas, while also serving as a symbol of change and progress. Establishing clean and justified district borders to substitute the existing convoluted and outdated traditional borders can create a positive aesthetic effect. They would become a visual representation of the country's evolution, and establish an association with the use of technology and rational planning.

Transitioning to a new administrative structure requires legislative adjustments, upfront relocation costs, and navigating potential resistance. Nevertheless, the long-term benefits for Latvian citizens are substantial. A more efficient and equitable administrative structure could improve the access to services for residents and contribute to the overall development of the country.

\section{Conclusion}

We propose a novel approach to administrative-territorial division based on the concept of the radius of edge-weighted graphs. The feasibility of computational and visualization methods has been demonstrated in the case of Latvia. The proposed optimization approach has been shown to have the potential to significantly improve the efficiency of the administrative structure in Latvia by reducing the number of TUs by 58\% while preserving the maximal travel time to the TU center. We encourage policymakers and relevant professionals to consider this innovative solution.

In future developments, additional edge weights and vertex weights can be added to the model to capture more road network and territorial characteristics. Other characteristics such as population-weighted distance, total distance, and distance weighted by socio-economic indicators, can be used in future iterations. 




%
\received{February 1, 2024}{*}{April 12, 2024}
\end{document}